\documentclass{eptcs}
  \usepackage{graphicx}
  \newcommand{\dom}{\mathop{\rm dom}\nolimits}   
  \newcommand{\ran}{\mathop{\rm ran}\nolimits}   
  
  \newcommand{\pws}{\mathop{\mathbf{P}}\nolimits}          
  \newcommand{\preco}[1]{\mbox{${}^\bullet{#1}$}} 
  \newcommand{\postc}[1]{\mbox{${#1}^\bullet$}}   
  \newcommand{\li}{\mathrel{\mathbf{li}}}
  \newcommand{\co}{\mathrel{\mathbf{co}}}
  \newcommand{\id}{\mathrel{\mathbf{id}}}
  \newcommand{\comp}{\mathrel{\smash{\leftrightarrow}}} 
  \newcommand{\cuts}[1]{\mathcal{C}(#1)}         
  \newcommand{\lines}[1]{\mathcal{L}(#1)}        
  \newcommand{\CC}[1]{\Gamma(#1)}            
  
  \newenvironment{proof}{\textit{Proof.}\quad}%
                        {\par\vspace{\baselineskip}\relax}
  \newtheorem{definition}{Definition}
  \newtheorem{proposition}{Proposition}
  \newtheorem{theorem}{Theorem}
  \newtheorem{lemma}{Lemma}

  \title{Between quantum logic and concurrency}
  \author{Luca Bernardinello
    \institute{%
          DISCo, Universit\`a degli studi di Milano--Bicocca\\
          viale Sarca 336 U14, Milano, Italia
    }
  \and
          Carlo Ferigato
    \institute{%
          JRC, Joint Research Centre of the European Commission\\
          via E. Fermi, 1 21027 Ispra, Italia
    }
  \and
          Lucia Pomello
    \institute{%
          DISCo, Universit\`a degli studi di Milano--Bicocca\\
          viale Sarca 336 U14, Milano, Italia
    }
  }
  \def\titlerunning{Quantum logic and concurrency}
  \def\authorrunning{L.Bernardinello, C. Ferigato, L. Pomello}
\begin{document}
    \maketitle
\begin{abstract}
  We start from two closure operators defined on the elements of
  a special kind of partially ordered sets, called causal nets.
  Causal nets are used to model histories of concurrent
  processes, recording occurrences of local states and of events.
  If every maximal chain (line) of such a partially ordered set meets
  every maximal antichain (cut), then the two closure operators coincide,
  and generate a complete orthomodular lattice. In this paper we
  recall that, for any closed set in this lattice, every line meets
  either it or its orthocomplement in the lattice, and show that
  to any line, a two-valued state on the lattice can be associated. 
  Starting from this result, we delineate a logical language whose
  formulas are interpreted over closed sets of a causal net, where
  every line induces an assignment of truth values to formulas.
  The resulting logic is non-classical; we show that maximal
  antichains in a causal net are associated to Boolean (hence
  ``classical'') substructures of the overall quantum logic.
\end{abstract}
\section{Introduction}
Partially ordered sets are a natural framework to model concurrent
processes, namely processes where
several components evolve in parallel, possibly interacting with each other.
In Petri net theory (see, for example, \cite{P77}, \cite{BF88})
the behaviour of concurrent
systems is modelled by causal nets,
a class of Petri nets which records a partial order between
occurrences of local states and events.

The partial order reflects relations of causal dependence, while
concurrency (or independence) is given by lack of
mutual order.
Petri axiomatized such partial orders, with the aim of describing
the flow of information in non-sequential processes, by drawing from
the laws of physics, and especially from the theory of relativity
(\cite{P_ijtp82}).
Among other properties, Petri introduced \emph{K-density}, which
requires that any maximal antichain (or \emph{cut}) in the partial order
and any maximal chain (or \emph{line}) have a non-empty intersection.
A line can be interpreted as the history of a sequential subprocess
(a particle, a signal), while cuts correspond to time instants,
where time is to be intended in a way analogous to the time
coordinate in relativistic spacetime.
K-density then requires that, at any instant, any sequential
subprocess must be in a definite local state or involved in a
change of state, modelled by an event (see~\cite{BF88}).

Two binary relations, corresponding to causal
dependence and to concurrency, can be defined on the elements
of an occurrence net. This suggests an analogy with relativistic
spacetime, and in particular with Minkowski spacetime. Several
authors have studied algebraic structures derived from Minkowski
spacetime (in particular, see~\cite{C02, CJ77}).
The main result related to our work consists in defining a lattice
whose elements are special sets of points in Minkowski, or even in
more general, spacetimes, and in showing that such a lattice is
complete and orthomodular.

Following similar ideas, but working on discrete partial orders,
in~\cite{BPR10} a closure operator is defined, whose closed sets
are subsets of points of the partial order. It is furthermore
shown that under a weak form of K-density, which holds for Petri
causal nets, the closed sets form a complete orthomodular lattice.
In the same work it is shown that, in a discrete and locally finite
framework, K-dense posets are exactly those in which, chosen an
arbitrary closed set, any line intersects either it or its
orthocomplement.

In this paper, we strengthen this last result by showing that
each line identifies a two-valued state, in the sense of quantum
logic. This suggests to look at the closed sets as propositions
in a logical language (see Section~\ref{s:logica}),
where orthocomplementation corresponds
to negation, so that any line induces an interpretation.
The resulting logical framework is obviously non-classical.

While lines (or maximal chains) correspond, in this sense, to
two-valued states in a quantum logic, cuts (or maximal antichains)
correspond to Boolean substructures of the logic, as stated in
Section~\ref{s:results}.

In the next section, we recall definitions and properties to
be used later, related to quantum logics, closure operators,
and Petri nets.
Section~\ref{s:results} collects our main results.
\section{Preliminaries}\label{s:preldef}
In this section, we recall some basic definitions and results,
which will be used in the rest of the paper, related to quantum logics,
Petri nets, and closure operators.
\subsection{Quantum logic}
The main reference for this section is~\cite{PP91}. A useful
treatment of quantum logic, also in relation to an alternative
formulation based on partial Boolean algebras, can be found
in~\cite{H89}.
\begin{definition} \cite{PP91}
  A \emph{quantum logic} $(P, \leq, 0,
    1,(.)')$ is a partially ordered set $(P, \leq)$, equipped with
    a minimum element, denoted by 0, and a maximum element, denoted
  by 1, and with a map
  $(.)':P \rightarrow P$ --- called
  \emph{orthocomplementation} --- such that the following
  conditions are satisfied (where $\lor$ and $\land$ denote,
  respectively, the least upper bound and the greatest lower bound
  with respect to $ \leq $, when they exist):
  $\forall x, y \in P$
  \begin{enumerate}
    \item  $(x')'  = x$
    \item  $x \leq y    \Rightarrow y' \leq x'$
    \item  $x \land x'  = 0$ and $x \lor x'  = 1$
  \end{enumerate}
  Two elements $x, y \in P$ are \emph{orthogonal}, denoted $x\, \perp
  \, y$, if $x \leq y'$. $P$ is \emph{orthocomplete} when
  every countable subset of pairwise orthogonal elements of $P$ has a least upper
  bound. Moreover, $P$ is \emph{orthomodular} when the
  following condition, called \emph{orthomodular law},
  \[
    x \leq y\quad  \Rightarrow\quad y = x \lor (y \land x')
  \]
  is satisfied.
\end{definition}
Figure~\ref{f:orthoql} shows one finite and one infinite quantum logic.
\begin{figure}
  \begin{center}
\includegraphics[width=0.3\linewidth]{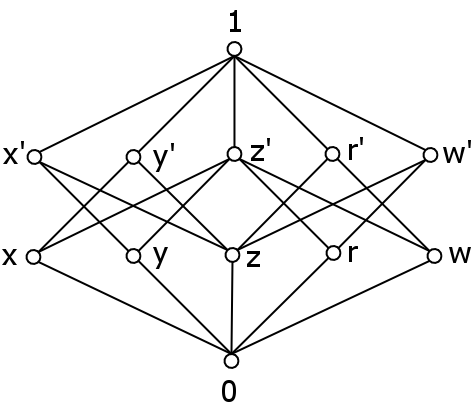}
\quad\quad
\includegraphics[width=0.46\linewidth]{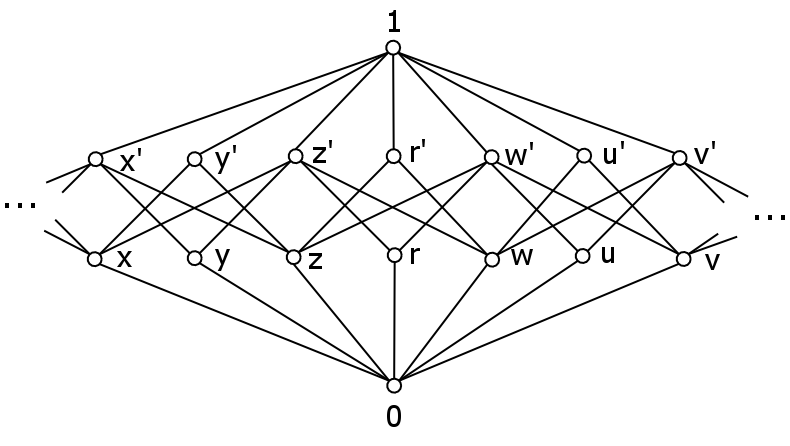}
  \end{center}
  \caption{Examples of orthomodular quantum logics.}\label{f:orthoql}
\end{figure}
In a quantum logic, $\perp $ is a symmetric relation and the De
Morgan laws hold: 
$(x \lor y)'= x' \land y'$, $(x \land y)'= x' \lor y'$
whenever one of the members of the equation exists.
We will sometimes use \emph{meet} and
\emph{join} to denote, respectively, $\land$ and $\lor$ with the
obvious extension to families of elements, denoted by $\bigwedge$ and
$\bigvee$. In the following we will sometimes use \emph{logic} as a shorthand
for \emph{quantum logic}.
%
\begin{definition}
Two elements $x, y$ of a logic $P$ are said to be
\emph{compatible} --- denoted by $x\comp y$ --- if there exist in
$P$ three mutually orthogonal elements $x_1, z$ and $y_1$
such that $x=x_1\lor z$ and $y=y_1\lor z$.
\end{definition}
\begin{definition}
A logic $P$ is called \emph{regular} if for any set
$\{x, y, z\}$ of pairwise compatible elements in $P$,
it holds that $x\comp (y \lor z)$.
\end{definition}

Every logic whose supporting poset is a lattice is regular
(\cite{PP91}, proposition 1.3.27).

A regular quantum logic admits an alternative characterization,
as a \emph{partial Boolean algebra} (see~\cite{H89}).
We will not give the formal definitions related to this view,
but only recall it informally.
A partial Boolean algebra is a family of partially overlapping
Boolean algebras, which share the minimum and the maximum
element, and satisfy a set of axioms on the shared elements.

\begin{definition} \cite{PP91}
A \emph{two-valued state} on a logic $P$ is a map
$s:P \rightarrow \{0, 1\}$
such that, for any sequence $(a_i)_{i \in I}$
of mutually orthogonal elements in $P$:
\begin{enumerate}
  \item $s(1) = 1$
  \item $s\big(\bigvee_{i\in I} a_i\big) = \Sigma_{i\in I}s(a_i)$
\end{enumerate}
\end{definition}

The set of two-valued states on a logic $P$ will be called
$\mathcal{S}_2(P)$.
If the elements of a quantum logic are interpreted as propositions
of a logical language, two-valued states can be considered as
consistent assignments of truth values to those propositions.
\subsection{Petri nets and causal nets}\label{s:nets}
A Petri net is a discrete model of a concurrent system, based on
the notions of local state (or condition) and of local change of
state (or event). Formally, they are bipartite graphs, where the
arcs encode immediate causal relations.

Here, we will focus on a special case of Petri nets, where the
underlying graph induces a partial order on the set of local states
and events.
\begin{definition} \label{d:net}
  A \emph{net} is a triple $ N = (B, E, F) $, where
  $B$ and $E$ are countable sets,
  $ F \subseteq ( B \times E ) \cup ( E \times B) $, and
  \begin{enumerate}
    \item $ B \cap E = \emptyset $;
    \item $ \dom(F) \cup \ran(F) =  B \cup E$.
  \end{enumerate}
\end{definition}

The elements of $B$ are called \emph{conditions}, the elements of $E$
are called \emph{events} and $F$ is called \emph{flow relation}.  

The standard graphical notation for nets represents conditions as
circles, events as squares and the flow relation as directed arcs.

The intuitive meaning associated to conditions and events is that
conditions represent \emph{local states} of a system while events
represent \emph{local changes of state}.

For each $x \in B \cup E$, define
  $ \preco{x} = \{ y \in B \cup E \ | \ (y, x) \in F \} $,
  $ \postc{x} = \{ y \in B \cup E \ | \ (x, y) \in F \} $.
For $e \in E$, an element $b \in B$ is a \emph{precondition} of 
$e$ if $b \in \preco e$; it is a \emph{postcondition} of $e$
if $b \in \postc e$.
In the basic model of Petri nets, which we refer to in this
paper, a condition is either \emph{true} or \emph{false};
an event can \emph{fire} (happen) if its preconditions are
all true and its postconditions are false; the effect of
an event firing consists in making its preconditions false
and its postconditions true.

A net $N = (B, E, F)$ is \emph{simple} if for each $x, y \in B \cup E$:
$(\preco{x} = \preco{y}$ and $\postc{x} = \postc{y}) \Rightarrow x=y$. 

By $F^+$ we denote the irreflexive, transitive closure of $F$,
by $F^*$ we denote $F^+ \cup \id_X$.

\begin{definition}\label{d:causal_net}
A \emph{causal net} is a net in which the following conditions hold:
    \begin{enumerate}
      \item $\forall b \in B: | \{ e \in E \mid
            (e, b) \in F \}| \le 1 \ \land \ | \{ e \in E \ | \ (b,
            e) \in F \}  | \le 1$;
      \item $\forall x,y \in B \cup E:(x,y) \in F^+
        \Rightarrow (y,x) \not= F^+$.
    \end{enumerate}
\end{definition}

A causal net is a net without cycles, and such that branches can occur
only at events. In the standard terminology, no \emph{choices} are
allowed.

The structure $(X,\sqsubseteq)$ derived from a causal net $N$ 
by putting $X=B \cup E$ and $\sqsubseteq=F^*$ is a partially ordered
set (shortly a \emph{poset}). A subset $S \subseteq X$ is
said to be \emph{convex} if, for each pair of its elements, $S$
contains the interval between them:
$\forall x, y \in S:$ $[x,y] \subseteq S$ where $[x,y]=\{z \in X \ |
\ x \sqsubseteq z \sqsubseteq y\}$.

\begin{definition}\label{d:local_finiteness}
$(X, \sqsubseteq)$ is \emph{interval-finite}
$\Leftrightarrow \forall x,y \in X: |[x,y]|< \infty$.
$(X, \sqsubseteq)$ is \emph{degree-finite}
$\Leftrightarrow \forall x \in X:  
|\preco x| < \infty \textrm{ and } |\postc x| < \infty$. When $(X,
\sqsubseteq)$ is both interval and degree-finite, we will say that
it is \emph{locally finite}. We will apply these terms also to the causal
net from which $(X, \sqsubseteq)$ is obtained.\\
\end{definition}
When using discrete partial orders to model processes of
real systems, local finiteness is a natural
assumption since
synchronization of infinite events or
causal dependence at infinite distance between the events cannot be
realized.

On the poset $(X,\sqsubseteq)$, the relations
$\li\ =\ \sqsubseteq \cup \sqsubseteq^{-1}$,
and $\co\ = (X \times X)\setminus \li$ can be
defined. Intuitively, $x \li y$ means that $x$ and $y$ are connected
by a causal relation while $x \co y$ means that $x$ and $y$ are
causally independent.
The relations $\li$ and $\co$ are symmetric and not transitive. 
Moreover, $\li$ is a reflexive relation, while $\co$ is irreflexive.

Given an element $x \in X$ and a set $S \subseteq X$, 
we write $ \ x \co S$  if $\forall y \in S: x \co y$.
Given two sets $S_1 \subseteq X$ and $S_2 \subseteq X$, 
we write $S_1 \co S_2$ if
$\forall x \in S_1, \forall y \in S_2: x \co y$.

A \emph{clique} of a binary relation is a set of pairwise related elements;
a clique of $\co \cup \id_X$ will be also called a \emph{coset};
a coset made of conditions only will be called a \emph{B-coset}.
Maximal cliques of $\co \cup \id_X$ and $\li$ are called, respectively,
\emph{cuts} and \emph{lines}. Given a poset $(X,\sqsubseteq)$, its
cuts and lines will be denoted, respectively, as:\\
\[
\cuts{X} =
     \{c \subseteq X \ | \ c \ \textrm{is a maximal clique of } \co
     \cup \id_X \}
\]
  and
\[
\lines{X} =
     \{\ l \subseteq X \ | \ l \ \textrm{is a maximal clique of } \li\}.
\]
We assume the axiom of choice, so any clique
of $\co \cup \id_X$ and of $\li$ can be extended to a maximal
clique. This will be used in particular in the proof of Theorem
\ref{t:line_2vs} below.

In the following, $(X,\sqsubseteq)$ will always be the poset obtained
from a causal net $N$, we will indicate the cuts and lines in $X$
as well as $\cuts{N}$ and $\lines{N}$. When a cut is
composed exclusively of conditions, it will be referred to as a
\emph{B-cut}. If $A\subseteq X$, the definition is restricted
naturally to the maximal cliques in $A$ with the notation $\cuts{A}$
and $\lines{A}$.

\begin{definition}\label{d:K-density}
$(X,\sqsubseteq)$ is \emph{K-dense} $\Leftrightarrow 
\forall c \in \cuts{N}, \forall l \in \lines{N}: c \cap l
\not= \emptyset$.
\end{definition}
From their definition, it follows immediately that, if a line
and a cut have a non-empty intersection, then the intersection
will consist in exactly one point. When this is the case, we
will say that the line and the cut \emph{meet} at a point, or
that the line \emph{crosses} the cut, or viceversa.
%
\subsection{Closure operators on causal nets}\label{s:causal_closure}
%
In this section, we recall the definition and some basic
properties of two closure operators on the set of elements
of a partially ordered set and, more specifically, on the
set of elements of a causal net.

In general, by closure operator on a set $Z$, we mean a map
$\gamma: \pws(Z) \rightarrow \pws(Z)$ (where $\pws(Z)$ denotes
the powerset of $Z$), satisfying the following,
for all $A, B \subseteq Z$:
\begin{enumerate}
  \item $A \subseteq \gamma(A)$ (increasing)
  \item if $A \subseteq B$, then $\gamma(A) \subseteq \gamma(B)$ (monotone)
  \item $\gamma(\gamma(A)) = \gamma(A)$ (idempotent)
\end{enumerate}
A subset $A$ of $Z$ is called \emph{closed} with respect to $\gamma$
if $A = \gamma(A)$.

The first operator with which we will deal here is defined
indirectly, starting from the
definition of \emph{causally closed sets}. These form a family of sets
closed by intersection, and with a maximum, with respect to set
inclusion. As usual, the closure of an arbitrary set
$A \subseteq B \cup E$ is by definition
the intersection of all the causally closed sets that contain $A$.

Let $N = (B, E, F)$ be a locally-finite causal net,
and $X = B \cup E$.
A subset of $X$ is causally closed if it is convex with respect to the
partial order induced by $N$, and if it is closed with respect
to local causes (preconditions) and local effects (postconditions).
\begin{definition} \label{d:sottoinsiemi_causalmente_chiusi}
A set $C \subseteq X = B \cup E$ is a \emph{causally closed set} if
\begin{itemize}
\item [\emph{(i)}] $\forall e \in E$, $\preco e \subseteq C \Rightarrow
e \in C$,
\item [\emph{(ii)}] $\forall e \in E$, $\postc e \subseteq C \Rightarrow
e \in C$,
\item [\emph{(iii)}] $\forall e \in E$, $e \in C \Rightarrow
\preco e  \cup \postc e \subseteq C$,
\item [\emph{(iv)}] $\forall x,y \in C$, $x \li y \Rightarrow
[x,y] \subseteq C$.
\end{itemize}
The family of causally closed sets of $N$ will be called $\CC{N}$.
\end{definition}
Define the \emph{border} of a subset $A \subseteq X$ as those
elements of $A$ which are directly linked, by F-arcs, to elements outside
of $A$:
$\beta(A) = \{ x \in A \mid
    \exists y \in X\setminus A: (x,y) \in F \cup F^{-1} \}$.
As shown in~\cite{BPR10}, the border of a causally closed set is
made of local states only:
\[
  \forall A \in \CC{N}\quad \beta(A) \subseteq B
\]
Moreover, $\CC{N}$ is closed by
intersection and $\emptyset \in \CC{N}$ and $B \cup E \in
\CC{N}$.
Hence, the family $\CC{N}$  forms a complete lattice --- non
orthocomplemented in the general case --- where meet is given by set
intersection and join is given by the causal closure of
the set union of the operands.
The associated closure operator, here denoted by $\phi$,
can now be defined as usual.
Let $X = B \cup E$.
\begin{definition}\label{d:costruzione_sottoinsiemi_causalmente_chiusi}
Define $\phi: \pws(X) \rightarrow \pws(X)$ as follows:
$\forall A \subseteq X$,
$\phi(A)=\bigcap \{C_i \ | \ C_i \in \CC{N}$ and $A \subseteq C_i\}$.
\end{definition}
\begin{figure}
  \begin{center}
\includegraphics[width=0.3\linewidth]{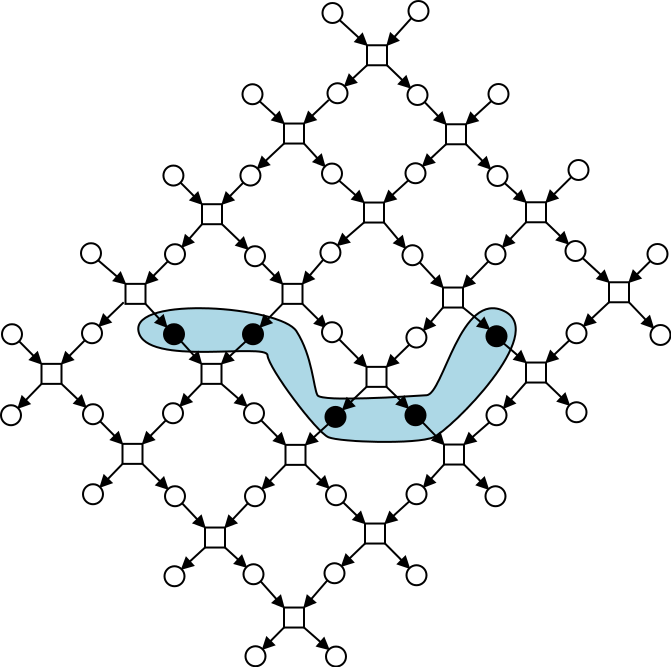}
\quad
\includegraphics[width=0.3\linewidth]{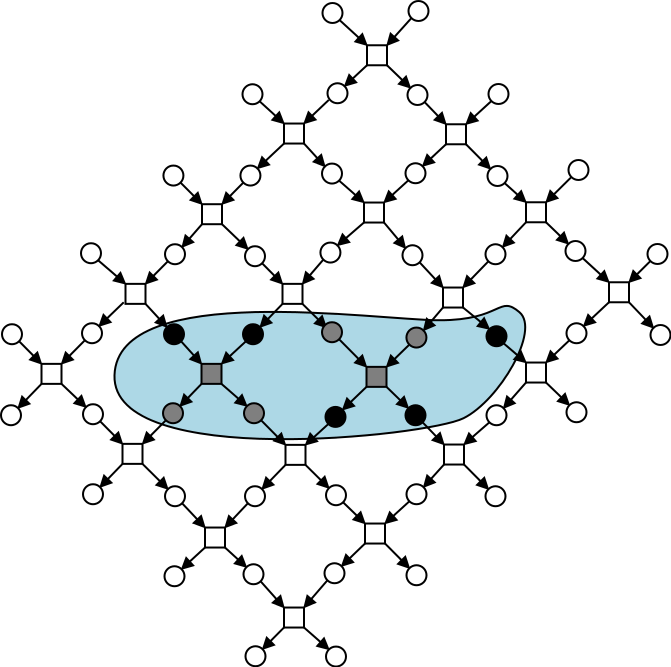}
\quad
\includegraphics[width=0.3\linewidth]{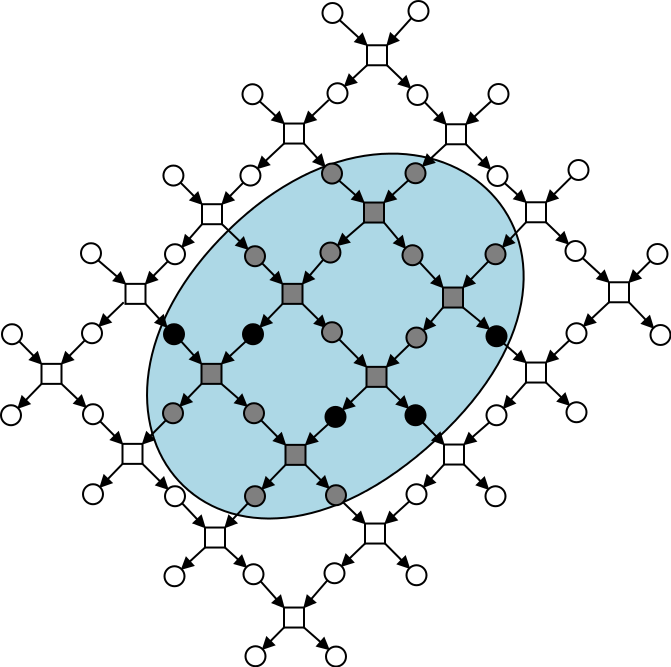}
  \end{center}
  \caption{Iterative closure of a B-coset.}\label{f:it_clos}
\end{figure}
So, given two elements $A$ and $B$ in $\CC{N}$, $A\vee B$ is defined
to be $\phi(A\cup B)$.
As shown in~\cite{BPR10}, the causal closure of a B-coset can be
defined by means of an iterative procedure, which adds new
elements according to the axioms in
Definition~\ref{d:sottoinsiemi_causalmente_chiusi}.
An example is shown in Figure~\ref{f:it_clos}.
The picture on the left shows a B-coset; the middle picture
shows an intermediate step in the procedure, while the picture
on the right shows the final step (other intermediate steps
are omitted).

The second closure operator is defined following a well-known
construction, to be found in \cite{B79}, and recalled below.

Given a set $Z$, a symmetric relation $\alpha \subseteq Z \times Z$
and a subset $A$ of $Z$, define
\begin{displaymath} 
  A'=\{x \in Z \ | \ \forall y \in A: (x,y) \in \alpha\}.
\end{displaymath}
By applying twice this operator,
we get a new operator $(.)''$, which can be shown to be
a closure operator.
A subset $A$ of $Z$ is \emph{closed}  
if $A = A''$. The family $L(Z)$ of all closed sets of $Z$, ordered
by set inclusion, is then a complete lattice.
When $\alpha$ is also irreflexive, 
the operator $(.)'$, applied to elements of $L(Z)$,
is an orthocomplementation; the structure
$\mathbf{L} = (L(Z), \subseteq, \emptyset, Z, (.)')$ 
then forms an orthocomplemented complete lattice \cite{B79}. 
\begin{figure}
  \begin{center}
\includegraphics[width=0.4\linewidth]{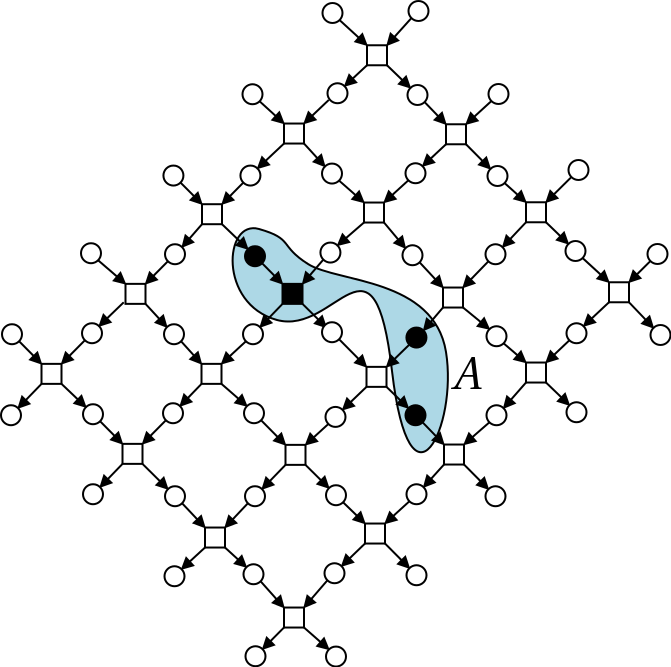}
\includegraphics[width=0.4\linewidth]{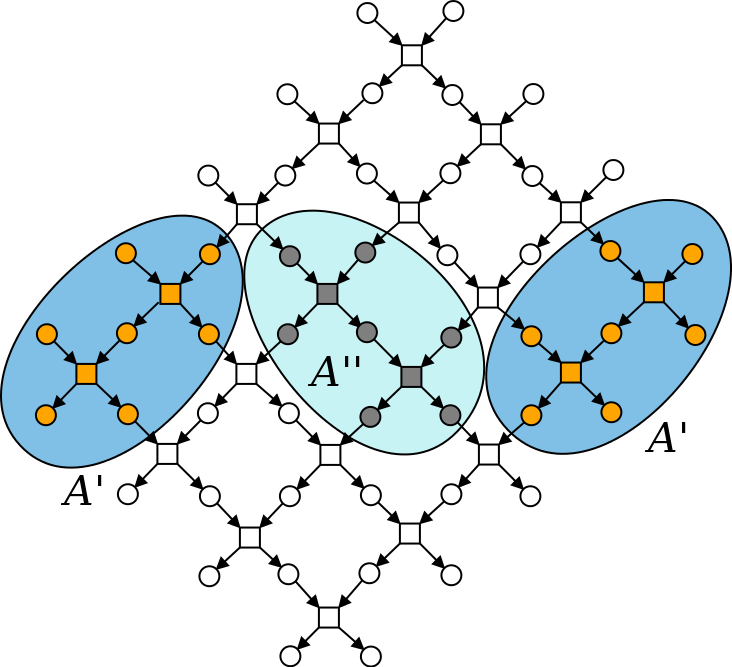}
  \end{center}
  \caption{A set $A$, its orthocomplement $A'$,
           and its closure $A''$.}\label{f:setclosure}
\end{figure}
Let us now consider the poset $(X,\sqsubseteq)$ derived from a causal net
$N$ together with the irreflexive relation $\co$ as in Section
\ref{s:nets}. By applying the
construction above, with $\co$ as irreflexive and symmetric
relation, to the subsets of $X$, we
obtain the orthocomplemented complete lattice
$(L(N), \subseteq, \emptyset, X, (.)')$, where $L(N)$ denotes the family
of closed subsets of $X$.

Figure~\ref{f:setclosure} shows a set $A$ of elements of
a causal net, together with its orthocomplement, $A'$,
and its closure, $A''$.

In~\cite{BPR10} it is shown that this lattice is orthomodular if the
construction above is applied to a locally finite poset satisfying
a weak form of
K-density, called N-density. All causal nets are N-dense (\cite{BF88}),
hence the lattice of closed sets derived from the concurrency relation
in a locally finite causal net is complete orthomodular.

As already noted, in general, the lattice of causally closed sets of a locally finite
causal net is not even orthocomplemented; however, if the causal net
is K-dense, then the two closure operators coincide, and the lattice
of causally closed sets is complete orthomodular.
\begin{theorem} \label{t:chiusieCC}\cite{BPR10}
Let $N = (B, E, F)$ be a locally finite, K-dense causal net.
Let  $A \subseteq B \cup E$.
Then $A \in \CC{N}  \iff  A \in L(N)$.
\end{theorem}
%
%
\section{Towards a logical view of closed sets}\label{s:results}
%
In this section, we collect the main results of our
contribution. The first states that every line in a
K-dense causal net $N$ identifies a two-valued state in
the lattice (or quantum logic) of closed sets, $L(N)$.
This suggests to look at the closed sets as propositions
in a logical language (see Section~\ref{s:logica}).
The second result concerns the relation between B-cuts
and Boolean subalgebras of a quantum logic.

Throughout this section, $N = (B, E, F)$ will denote an
arbitrary locally finite, K-dense causal net.

The following lemma states that any closed set can be
obtained as the closure of any of its local B-cuts.
\begin{lemma}\label{l:tagliechiusi}
  Let $A \in L(N)$, and $\tau$ a B-cut of $A$.
  Then $A = \tau''$.
\end{lemma}
\begin{proof}
  The closure operator is idempotent and monotone;
  hence, from $\tau \subseteq A$,
  we get $\tau'' \subseteq A'' = A$. To show the inclusion in the
  other direction, take $x \in A$. If $x \in \tau$, then $x \in \tau''$.
  Suppose $x \not\in \tau$. Since $\tau$ is a B-cut, $x \li z$ for some
  $z \in \tau$. By way of contradiction, suppose $x \not\in \tau''$;
  then, there must be $w \in \tau'$, with $x \li w$.

  Put $x < z$ (the symmetric case is treated analogously); then $x < w$,
  because $z \co w$. Choose a path from $x$ to $w$; this path must cross
  the border of $A$ at an $F$-arc from a condition $b$ to an event $e$.
  Then $b$ must be concurrent to all the elements of $\tau$, because $A$ is convex,
  but this contradicts the hypothesis that $\tau$ is a cut (maximal
  antichain) of $A$.
\end{proof}
The previous lemma implies that, for each $x \in B \cup E$,
if $x \co \tau$, then $x \co A$. This will be used in later proofs.

Now we can prove a crucial relation between lines and closed sets.
This is actually a corollary of Theorem~3.2 in~\cite{BPR10}, but
we think that a direct proof could be useful. It says that, given
a closed set $A$, a line crosses either $A$ or $A'$ (but not both).
The statement is illustrated in Figure~\ref{f:linea_chiuso}, where
a line is shown with a thicker stroke.
\begin{figure}
  \begin{center}
\includegraphics[width=0.5\linewidth]{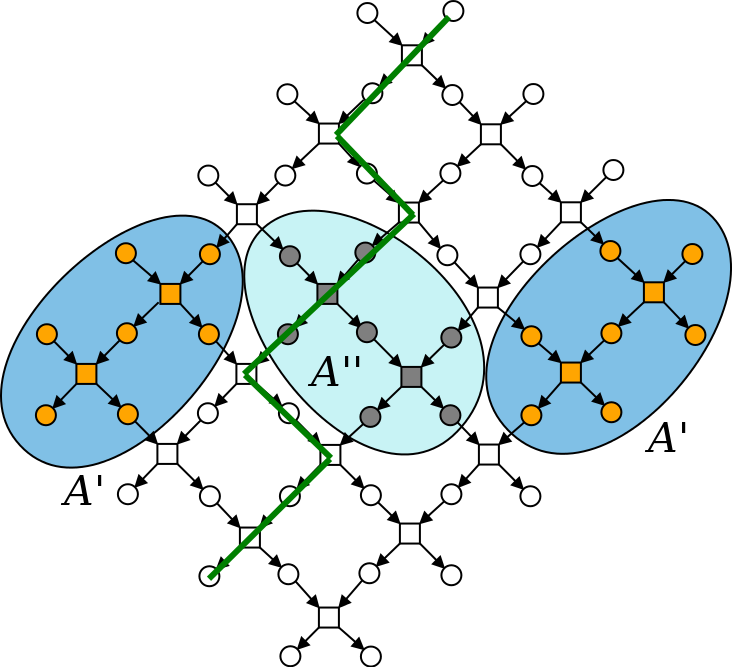}
  \end{center}
  \caption{A line crosses either a closed set or its orthocomplement.}%
  \label{f:linea_chiuso}
\end{figure}
\begin{proposition}\label{p:lmeetsc}
  Let $A \in L(N)$, and $\lambda \in \lines{N}$.
  Then
  \[
      \lambda \cap A \not= \emptyset
              \Leftrightarrow
      \lambda \cap A' = \emptyset
  \]
\end{proposition}
\begin{proof}
  Suppose $\lambda \cap A \not= \emptyset$. Any element in $A'$ is
  concurrent with any element in $A$. Since $\lambda$ is a clique
  of the $\li$ relation, no element in $\lambda$ can be in $A'$.

  Suppose now that $\lambda \cap A = \emptyset$. Take a B-cut of $A$,
  say $\tau_1$. This is a B-coset of $N$, and can be extended to a
  B-cut of $N$, say $\tau$. Since $N$ is K-dense, $\tau$ crosses
  $\lambda$ at a point $b \in B$. Then $b \co \tau_1$, and,
  by Lemma~\ref{l:tagliechiusi}, $b \co A$, which means $b \in A'$.
\end{proof}
Building on the previous proposition, we now define a map associated to
a line in $N$. The map can be seen as the characteristic map of
the family of closed sets that cross the given line.
\begin{definition}
  Let $\lambda$ be a line of $N$. Define
  $\Delta(\lambda) = \{ A \in L(N) \mid A \cap \lambda \not= \emptyset \}$,
  and $\delta_{\lambda}: L(N) \rightarrow \{0, 1 \}$ in this way:
  for each $A \in L(N)$, $\delta_{\lambda}(A) = 1$ if
  $A \in \Delta(\lambda)$, 0 otherwise.
\end{definition}
\begin{theorem}\label{t:line_2vs}
  The map $\delta_{\lambda}$ is a two-valued state of $L(N)$.
\end{theorem}
\begin{proof}
  Let $(A_i)_{i \in I}$ be a family of pairwise orthogonal closed sets.
  Then, for each $i, j, i\not= j$, $A_i \subseteq A_j'$. Hence, if
  $\lambda$ crosses one of the $A_i$s, it cannot cross any of the
  others because of Proposition \ref{p:lmeetsc}.
  From this, it follows that,
  if $\delta_{\lambda}(A_i) = 1$, then $\delta_{\lambda}(A_j) = 0$
  for any $j \not= i$, and $\delta_{\lambda}(\bigvee A_i) = 1$.

  Suppose now that $\delta_{\lambda}(A_i) = 0$ for each $i \in I$.
  We will show that $\lambda$ does not cross $\bigvee A_i$.
  For each $i \in I$, take a B-cut of $A_i$. The elements of all the
  B-cuts of the $A_i$'s are pairwise
  concurrent; hence their union is a B-coset, and can be extended to
  a B-cut of $N$. Since $N$ is K-dense, this cut crosses $\lambda$ at
  a point, say $s$, and $s$ is concurrent with all the $A_i$s.
  By way of contradiction, suppose now that
  $\bigvee A_i$ crosses $\lambda$ at a point $b$. Since this point is
  in the closure of $\bigcup A_i$, it must satisfy $b \co (\bigvee A_i)'$.
  But this implies $b \co s$, contradicting the hypothesis that $b$
  and $s$ are on the same line.
\end{proof}
The previous results relate lines and two-valued states in
quantum logics. We now briefly point out a dual relation
between cuts in a causal net and Boolean subalgebras in the
lattice of closed sets.

We have already noted that a regular quantum logic can be
seen as a family of partially overlapping Boolean algebras.
In our context, these component Boolean algebras are atomic.

Consider a B-cut of $N$, say $\tau = \{b_1, \ldots, b_i, \ldots \}$.
For each $b_i$ in $\tau$, compute $\{ b_i \}'' = \beta_i$.
The closed sets obtained in this way are pairwise orthogonal
in the lattice of closed sets, which means that they are
pairwise concurrent in $N$. Their join gives $B \cup E$.

Then, the set $\{ \beta_1, \ldots, \beta_i, \ldots \}$ is the
set of atoms of a (maximal) Boolean subalgebra of $L(N)$.
More generally, given a family $(A_i)_{i \in I}$ of pairwise
orthogonal (concurrent) closed sets, such that
$\bigvee_{i \in I} A_i = B \cup E$, there is a Boolean
subalgebra of $L(N)$ such that the $A_i$s are its atoms.
%
\subsection{Logic}\label{s:logica}
The statement of Theorem~\ref{t:line_2vs} suggests to look at
the closed sets in $L(N)$ as propositions of a logical language,
which admits a different interpretation for each line in $N$.

In fact, each line selects exactly one closed set for each pair
$(A, A')$, and this selection is consistent with the structure
of the lattice of closed sets: let us say that the proposition
associated to $A$ is true, with respect to a line $\lambda$,
if $\lambda$ crosses $A$, and false otherwise.
Then, we can take $(.)'$ as a negation, while the lattice
operations correspond to the logical connectives, disjunction
(join), and conjunction (meet). Theorem~\ref{t:line_2vs}
guarantees that, under this interpretation, if two propositions
are true, then also their conjunction is true, while the
conjunction of two ``compatible'' propositions is true only
if at least one of them is true, where two propositions are
compatible if their corresponding closed sets are compatible
in the quantum logic $L(N)$ (or, equivalently, if there is
a Boolean subalgebra of $L(N)$ which contains both).

The resulting logic is obviously non-classical, since the lattice
of closed sets is not, in general, distributive. It is, so to
speak, locally classical, in the sense that the set of closed
sets associated to true propositions, projected on a Boolean
subalgebra of $L(N)$ gives an ultrafilter.

Formally, we define the propositional language ${\cal F}_\Pi$
and its interpretation over the orthomodular lattice $L(N)$.
Let $\Pi = (\pi_i)_{i\in I}$ be a set of \emph{propositions}.
Define the set ${\cal F}_\Pi$ of \emph{formulas} over $\Pi$, inductively,
as follows:
\begin{itemize}
  \item [\emph{(i)}] every $\pi_i$ is a formula;
  \item [\emph{(ii)}] if $f_1, f_2$ are formulas, then
        $f_1 \vee f_2$, $f_1 \wedge f_2$, $\neg f_1$, $f_1 \rightarrow f_2$
        are formulas;
  \item [\emph{(iii)}] nothing else is a formula.
\end{itemize}
An interpretation of the language of formulas ${\cal F}_\Pi$ is a pair
$J = \langle h: \Pi \rightarrow L(N), \lambda \in \lines{N} \rangle$.

To each formula $f$, we can associate an element of $L(N)$,
by defining a map $i:{\cal F}_\Pi \rightarrow L(N)$, as follows.
If $f = \pi_i$, then $i(f) = h(\pi_i)$; if $f = \neg(f_1)$, then
$i(f) = (i(f_1))'$;
if $f = f_1 \vee f_2$, then $i(f) = i(f_1) \vee i(f_2)$, where
$\vee$ is the join operation in the lattice $L(N)$;
if $f = f_1 \wedge f_2$, then $i(f) = i(f_1) \wedge i(f_2)$, where
$\wedge$ is the meet operation in the lattice $L(N)$;
if $f = f_1 \rightarrow f_2$, then
$i(f) = (i(f_1))' \vee i(f_2)$.

With this definition, the implication connective and the
partial order on $L(N)$ are related by the following
statement: if $i(f_1) \subseteq i(f_2)$, then
$i(f_1 \rightarrow f_2) = B \cup E$.

The choice of the line $\lambda$ in an interpretation
determines the assignment of
truth values to formulas: a formula is true if $\lambda$
crosses the closed set associated to the formula by the
map $i$.
Formally, we define a satisfiability relation between interpretations
and formulas:
\[
  J \models f \quad \Leftrightarrow \quad i(f) \cap \lambda \not= \emptyset.
\]
This definition is consistent, in the sense that, by direct
verification, one can prove the following:
\begin{enumerate}
  \item $J \models f \wedge g$ if, and only if, $J \models f$
        and $J \models g$
  \item $J \models \neg f$ if, and only if,
        $J \not\models f$
  \item $J \models f \vee g$ if, and only if,
        $i(f) \comp i(g)$, and either $J \models f$ or $J \models g$;
  \item $J \models f \rightarrow g$ if, and only if,
        $i(f) \subseteq i(g)$.
\end{enumerate}
The case of disjunctive formulas can be illustrated by means
of a simple example. Consider the net below
(which might be a fragment of a larger net).
\begin{center}
  \includegraphics[width=0.2\linewidth]{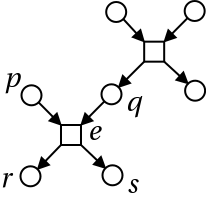}
\end{center}
Let $f$ and $g$ be two propositions interpreted,
respectively, over $\{p\}$ and $\{r\}$, which are
closed sets. Then
$i(f \vee g) = \{p, q, e, r, s\}$. With respect to a line
$\lambda$ containing $\{q, e, s\}$, the formula $f \vee g$ is true,
while $f$ and $g$ are false. In this
case $i(f)$ and $i(g)$ are not compatible in $L(N)$.

Suppose now to interpret $f$ and $g$ over $\{p\}$ and $\{q\}$,
which are compatible in $L(N)$. The interpretation of $f \vee g$
is the same as before, namely $\{p, q, e, r, s\}$, but any
line crossing this closed set is bound to cross either
$\{p\}$ or $\{q\}$, so that either $f$ or $g$ is true.
\section{Conclusion and prospects}
A first connection between orthomodular structures
and concurrency theory emerged by studying an
abstract notion of local state in automata and in
Petri nets (see~\cite{BFP03}).

Later, a different connection was found, which bears a more
direct relation with special relativity theory.
Several authors had previously shown that a closure operator, and a
corresponding orthomodular lattice, can be derived from 
the ``spacelike'' relation between points in Minkowski spacetime
(\cite{C02,CJ77}).

A similar construction was then applied to discrete partially
ordered sets modelling the history of concurrent, or
distributed, system. In the discrete case, the orthomodularity
of the resulting lattice of closed sets depends on a feature
of the partially ordered set which was called
\emph{N-density} by C.A. Petri (\cite{BPR10}).

In this context, the specific character of the causal nets
introduced by Petri is the distinction between synchronization events
and local properties changed by the occurrence of the events. With
this distinction, \emph{lines} can be interpreted as \emph{signals}
whose status is changed by the interaction event with other signals
and preserving their local status until another interaction
occurs. This in analogy with flows of particles in space whose mutual
interactions are \emph{collisions}.

For causal nets and in partial orders derived from causal nets,
density properties were studied mainly with the aim of providing
a sound set of axioms for the definition of the causal structures
representing concurrent processes of systems.
In this respect, \cite{BF88} gives a
comprehensive survey of the properties of partially ordered sets
related to causal nets. In particular, the relations between
N-density and K-density
are presented in the specific context of models of concurrent systems.

Starting from what we have presented here, we will undertake
several further steps: characterize those orthomodular lattices that can
be obtained as lattices of closed sets induced by concurrency;
building a causal net whose lattice of closed sets is
isomorphic to a given lattice; study lattices of closed sets
induced by concurrency in other classes of Petri nets;
give a meaning to the logical language here defined.
\section*{Acknowledgments}
This work was partially supported by MIUR and by
MIUR-PRIN 2010/2011 grant `Automi e Linguaggi Formali:
Aspetti Matematici e Applicativi', code H41J12000190001.
\bibliographystyle{eptcs}
\bibliography{bfp_qpl2012}
\end{document}